\documentclass{article}
\usepackage[utf8]{inputenc}
\usepackage{cite}
\usepackage{graphicx}
\usepackage[export]{adjustbox}
\usepackage{textcomp}
\usepackage{graphicx}
\usepackage{mathrsfs}
\usepackage{amsmath}
\usepackage{amsthm}
\usepackage{amssymb}
\usepackage{mathabx}
\usepackage{xspace}
\usepackage{hyperref}
\usepackage{cleveref}
\usepackage{fullpage}
\usepackage{listings}
\usepackage[disable]{todonotes}

\newenvironment{psmallmatrix}{\left(\begin{smallmatrix}}{\end{smallmatrix}\right)}

\newtheorem{theorem}{Theorem}[section]
\newtheorem{corollary}{Corollary}[theorem]
\newtheorem{lemma}[theorem]{Lemma}

\theoremstyle{definition}
\newtheorem{definition}{Definition}[section]

\usepackage{mathtools, stmaryrd}
\usepackage{xparse} 
\DeclarePairedDelimiterX{\Iintv}[1]{\llbracket}{\rrbracket}{\iintvargs{#1}}
\NewDocumentCommand{\iintvargs}{>{\SplitArgument{1}{,}}m}
{\iintvargsaux#1} %
\NewDocumentCommand{\iintvargsaux}{mm} {#1\mkern1.5mu..\mkern1.5mu#2}

\usepackage{array}
\newcolumntype{L}{>{$}p{.4\textwidth}<{$}} 
\newcolumntype{R}{>{$}p{.55\textwidth}<{$}}

\newcommand{\newclass}[2]{\newcommand{#1}{{\text{\upshape\sffamily #2}}\xspace}}
\renewcommand{\P}{{\text{\upshape\sffamily P}}\xspace}
\newclass{\NP}{NP}
\renewcommand{\L}{{\text{\upshape\sffamily L}}\xspace}
\newclass{\NL}{NL}
\newclass{\CL}{CL}
\newclass{\TIME}{TIME}
\newclass{\SPACE}{SPACE}
\newclass{\CSPACE}{CSPACE}
\newclass{\NCo}{NC$^1$}
\newclass{\SACo}{SAC$^1$}
\newclass{\ACo}{AC$^1$}
\newclass{\TCo}{TC$^1$}
\newclass{\NCt}{NC$^2$}
\newclass{\SACt}{SAC$^2$}
\newclass{\ShSACt}{\#SAC$^2$}
\newclass{\NC}{NC}
\newclass{\SAC}{SAC}
\newclass{\AC}{AC}
\newclass{\TC}{TC}
\newclass{\BPL}{BPL}
\newclass{\ZPP}{ZPP}

\title{Catalytic Computing and Register Programs Beyond Log-Depth}
\author{Yaroslav Alekseev\thanks{Supported by ISF grant 507/24.} \\ Technion Israel Institute of Technology \\ \texttt{tolstreg@gmail.com} \and Yuval Filmus\footnotemark[1] \\ Technion Israel Institute of Technology \\ \texttt{yuvalfi@cs.technion.ac.il} \and
Ian Mertz \thanks{Supported by the Grant Agency of the Czech Republic under the grant agreement no. 24-10306S and by the Center for Foundations of Contemporary Computer Science
(Charles Univ. project UNCE 24/SCI/008).} \\ Charles University \\ \texttt{iwmertz@iuuk.mff.cuni.cz} \and
Alexander Smal \\ JetBrains Research \\ \texttt{avsmal@gmail.com} \and
Antoine Vinciguerra\footnotemark[1] \\ Technion Israel Institute of Technology \\ \texttt{antoine.v@campus.technion.ac.il}}
\begin{document}

\maketitle
\thispagestyle{empty}
\begin{abstract}
In a seminal work, Buhrman et al.\ (STOC 2014) defined the class $\CSPACE(s,c)$ of problems solvable in space $s$ with an additional catalytic tape of size $c$, which is a tape whose initial content must be restored at the end of the computation. They showed that uniform $\TCo$ circuits are computable in catalytic logspace, i.e., $\CL=\CSPACE(O(\log{n}), 2^{O(\log{n})})$, thus giving strong evidence that catalytic space gives $\L$ strict additional power. Their study focuses on an arithmetic model called register programs, which has been a focal point in development since then.

Understanding $\CL$ remains a major open problem, as $\TCo$ remains the most powerful containment to date. In this work, we study the power of catalytic space and register programs to compute circuits of larger depth. Using register programs, we show that for every $\epsilon > 0$,
\begin{equation*}
    \SACt \subseteq \CSPACE\left(O\left(\frac{\log^2{n}}{\log\log{n}}\right), 2^{O(\log^{1+\epsilon} n)}\right)
\end{equation*}
This is an $O(\log \log n)$ factor improvement on the free space needed to compute $\SACt$, which can be accomplished with near-polynomial
catalytic space.

We also exhibit non-trivial register programs for matrix powering, which is a further step towards showing e.g.\ $\NCt \subseteq \CL$.
\end{abstract}
\clearpage
\pagenumbering{arabic}

\section{Introduction}
\subsection{Catalytic Computation}

In the realm of space-bounded computation, the \textit{catalytic computing} framework, first introduced by Buhrman et al.\cite{buhrman2014computing}, investigates the power of having additional storage which has to be restored to its original value at the end of computation.
A catalytic Turing machine is a space-bounded Turing machine with two read-write tapes: a standard work tape of size $s$ and an additional, dedicated catalytic tape of size $c$. However, the catalytic tape begins completely filled in with some memory $\tau$, and despite being available for use as work space, the catalytic tape must be restored to its initial state $\tau$ at the end of the computation. 

While $\CSPACE(s,c)$, the class of functions solvable with work space $s$ and catalytic space $c$, clearly sits between $\SPACE(s)$ and $\SPACE(s+c)$, na\"{i}vely it might seem that catalytic space should not improve the power of machines, as was conjectured in many previous works~\cite{cook2012pebbles,liu2013pebbling,edmonds2018hardness,iwama2018read}. Unexpectedly, \cite{buhrman2014computing} showed that \textit{catalytic logspace} ($\CL=\CSPACE(\log{n}, 2^{O(\log{n})})$), the catalytic analogue of the complexity class $\L$, contains the circuit class $\TCo$, which contains non-deterministic and randomized logspace ($\NL$ and $\BPL$, respectively) as well as functions, such as determinant, widely believed to be outside both. Thus they gave a strong case for the power of catalytic memory as a resource.

Building on this result, the catalytic approach and methods has seen growing interest. Many follow-up works have sought to understand variants and properties of catalytic space have been studied, such as non-deterministic and randomized analogues~\cite{buhrman2018catalytic,datta2020randomized,cook2024structure,koucky2025collapsing}, non-uniform catalytic models~\cite{potechin2016note,robere2022amortized,cook2022trading,cook2023tree}, robustness to errors in catalytic resetting~\cite{gupta2024lossy,folkertsma2024fully}, catalytic communication protocols~\cite{pynecatalytic}, and other variants~\cite{gupta2019unambiguous,bisoyi2022pure,bisoyi2024almost} (see surveys by Kouck\'{y}~\cite{koucky2016catalytic} and Mertz~\cite{mertz2023reusing} for more discussion of these and other works).

Furthermore, applying catalytic tools  to other complexity theoretic questions has seen successful results in recent years. Two of the most successful approaches have been 1) approaches to derandomizating non-catalytic space-bounded classes~\cite{pyne2024derandomizing,doron2024opening,li2024distinguishing}; and 2) space-bounded algorithms for function composition~\cite{cook2020catalytic,cook2021encodings,cook2023tree}, which recently culminated in a breakthrough simulation by Williams~\cite{williams2025} of time with quadratically less space.

Nevertheless, the exact strength of catalytic space remains open. In their first work, \cite{buhrman2014computing} showed $\TCo \subseteq \CL \subseteq \ZPP$, with the key open problem being the relationship of $\CL$ to $\P$. Further work has shown slight progress on both ends; Cook et al.~\cite{cook2024structure} showed that $\CL$ reduces to the \textit{lossy coding} problem, which itself is in $\ZPP$, while Agarwala and Mertz~\cite{agarwala2025matching} showed that bipartite matching, a function not known to be anywhere in the $\NC$ hierarchy, can be solved in $\CL$.

The key open question in this work is whether or not $\CL$ contains higher levels of the $\NC$ hierarchy. Mertz~\cite{mertz2023reusing} posed a concrete approach to showing $\NCt \subseteq \CL$ via \textit{register programs}, which we turn to now.

\subsection{Register Programs}

Register programs were first introduced by Coppersmith and Grossmann~\cite{coppersmith1975generators}, and revived later by Ben-Or and Cleve to generalize Barrington's theorem to arithmetic circuits~\cite{cleve1988computing,cleve1990methodologies}. A register program over a ring $\mathbf{R}$ and a set of inputs $x_1,\dots,x_n\in \mathbf{R}$ is defined as a sequence of instructions $I_1,\dots,I_n\colon \mathbf{R}^s\mapsto \mathbf{R}^s$ applied to a set of registers $\{R_1,\dots,R_s\}$.

The key technique of Buhrman et al.~\cite{buhrman2014computing} was based on register programs, which they showed can be directly simulated on a catalytic machine. They constructed a register program computing the $k$th power of some value $x$ with $k$ registers and $O(1)$ accesses to the value $x$; using some extensions and other works, they thus showing that $\CL$ contains the circuit class $\TCo$. This was also the driving force behind the result of Cook and Mertz~\cite{cook2023tree}, and by extension the work of Williams~\cite{williams2025}, who showed that the Tree Evaluation problem, which was suggested as a function which would separate $\L$ from $\P$~\cite{cook2012pebbles}, can be computed in space $O(\log{n}\log\log{n})$.

As with catalytic space more generally, the power of register programs are still a mystery, and deserves to be investigated thoroughly. The approach suggested in \cite{mertz2023reusing} for $\CL$ versus $\NCt$ is to find a register program computing the $k$th power of a \textit{matrix} with a register program matching the initial result of \cite{buhrman2014computing} in the non-commutative setting. More generally, it was proposed that further algebraic extension of these and other works could be the goal to improving the strength of $\CL$.

\subsection{Our Results}

In this work, we investigate the complexity of computing polynomials in the register program framework, and make the first progress towards catalytic algorithms for circuit classes beyond $\TCo$. We present register programs for different kinds of polynomials (such as symmetric polynomials and polynomials representing Boolean functions), as well as more efficient programs for evaluating non-constant depth $d$ boolean circuits with 
a constant number of recursive calls.

From this, we deduce that the class $\SACt$ of Boolean circuits of polynomial size and depth $O(\log^2{n})$, with bounded fan-in $AND$ gates and unbounded fan-in $OR$ gates, can be computed with $o(\log^2{n})$ work space when given access to nearly polynomial catalytic space.

\begin{theorem} \label{main_sac}
For all $\epsilon > 0$, 
\begin{equation*}
    \SACt\subseteq \CSPACE\left(O\left(\frac{\log^2{n}}{\log\log{n}}\right), 2^{O(\log^{1+\epsilon}{n})}\right)
\end{equation*}
\end{theorem}
\noindent
Our technique also gives such an improvement for $\NCt$ with only polynomial catalytic space, although this can be derived more directly from the results of \cite{buhrman2014computing}. It also extends to $\ShSACt$, the arithmetic variant of $\SACt$.

Second, we show sublinear register programs to compute matrix powering, thus making initial progress towards the program of \cite{mertz2020catalytic}:
\begin{theorem}
\label{Matrix_powering_intro}
   Let $d,n,p\in \mathbb{N}$, where $p$ is prime. Let $M\in M_n(\mathbb{F}_p)$. For all $\epsilon > 0$ there is a register program that computes $M^d$ with
    \begin{itemize}
        \item 
         $O_{\epsilon}( d^{\epsilon}\log{d})$ recursive calls to $M$
         \item 
         $O_{\epsilon}\left(n^{\exp(1/\epsilon)}\right)$ basic instructions, and
        \item
      $O_{\epsilon}\left(n^{\exp(1/\epsilon)}\right)$ registers over $\mathbb{F}_{O\left(\left(np\right)^{\exp(1/\epsilon)} \right)}$.
\end{itemize}
\end{theorem}
\section{Preliminaries}
\subsection{Circuits}

A Boolean circuit is a directed acyclic graph $C$ with $0,1$-valued inputs and outputs. Internal nodes (gates) are labeled with Boolean operations from a given set $B$. The circuit size $s(C)$ is the number of nodes, and its depth $d(C)$ is the longest input-output path. While Boolean circuits had been studied earlier~\cite{shannon1949synthesis}, their importance for parallel computing gained significant attention in the late 1970s~\cite{borodin1977relating,pippenger1979simultaneous,ruzzo1981uniform,borodin1983parallel}.

\begin{definition}
We define the following circuit families over input literals 

$x_1, \ldots, x_n, \lnot x_1, \ldots, \lnot x_n$:
\begin{itemize}
    \item 
    $\NC$ circuits: fan-in 2 AND and OR gates
    \item 
    $\SAC$ circuits: fan-in 2 AND gates and unbounded fan-in OR gates\footnote{Since $\SAC^i$ is closed under complement for all $i \geq 1$~(see~\cite{borodinCDRT89}), these fan-ins can also be reversed, but for our argument we specifically use this version.}
    \item 
    $\AC$ circuits: unbounded fan-in AND and OR gates
    \item 
    $\TC$ circuits: unbounded fan-in threshold gates
\end{itemize}
We denote by $\NC^i$ ($\SAC^i$, $\AC^i$, $\TC^i$) the family of functions computable by $\NC$ ($\SAC$, $\AC$, $\TC$) circuits of polynomial size and $O(\log^i n)$ depth. By $\NC$ we denote $\bigcup_{i \in \mathbb{N}} \NC^i$ and similarly for $\SAC$, $\AC$, and $\TC$.
\end{definition}

\noindent
The relations between these different classes has been extensively studied~\cite{ruzzo1981uniform,vollmer1999introduction}. For all $i$, we have the following relations:
\[\NC^i \subseteq \SAC^i \subseteq \AC^i  \subseteq \TC^i \subseteq \NC^{i+1}\]
As a consequence, we have that
\[\NC = \SAC = \AC = \TC\]
Furthermore, other lines of research~\cite{borodin1977relating,cook1983classification,venkateswaran1992circuit} have established other relationships:
\[ \NCo \subseteq \L  \subseteq \NL  \subseteq \SACo \subseteq \ACo \subseteq \TCo \subseteq \NCt \subseteq \SACt \subseteq \cdots \subseteq \NC \subseteq \P \]
While all containments are widely conjectured to be strict, no separations are known.

\subsection{Catalytic Computation}
Central to this work is the notion, introduced by Buhrman et al.~\cite{buhrman2014computing}, of catalytic space.

\begin{definition}
    A \emph{catalytic Turing machine} with work space $s := s(n)$ and catalytic space $c := c(n)$ is a space-bounded Turing machine $M$ with access to two read-write tapes: a standard work tape of size $s$, and a \emph{catalytic} tape of size $c$. The catalytic tape has the additional restriction that for every $\tau \in \{0,1\}^c$, if we run $M$ on any input $x$ with the catalytic tape in initial configuration $\tau$, the catalytic tape has final configuration $\tau$ as well.
\end{definition}

\noindent
This definition gives rise to natural complexity classes:

\begin{definition}
    We define $\CSPACE(s,c)$ as the class of problems solvable by a catalytic Turing machine with a work tape of size $s$ and a catalytic tape of size $c$. Furthermore, we denote \emph{catalytic logspace} to be $\CL = \CSPACE(O(\log n), 2^{O(\log n)})$.
\end{definition}

\subsection{Register Programs}

Inspired by Ben-Or and Cleve's work on straight-line programs~\cite{benor1992register,cleve1988computing},  Buhrman et al.~\cite{buhrman2014computing} investigated the potential for optimization, examining a more generic type of computational model due to Coppersmith and Grossman~\cite{coppersmith1975generators}.

\begin{definition}
    Let $\mathbf{R}$ be a ring. An \emph{$\mathbf{R}$-register program} over input $x_1,\dots,x_n$ with space $s$ is defined as a sequence of instructions $I_1,\dots,I_t$ applied to a set of registers $R_1,\dots,R_s$, where each instruction $I_k$
    has one of the following two forms:
    \begin{itemize}
       \item 
       \emph{Basic instruction}: updating a register $R_i$ with a polynomial $p_k$ over the other registers:
        \[
        I_k\colon R_i\leftarrow R_i + p_k(R_1,\dots,R_{i-1},R_{i+1},\dots,R_n)
        \] 
        \item 
        \emph{Input access}/\emph{Recursive call}: adds the input (scaled by some $\lambda\in \mathbf{R}$) to one register $R_i$:
        \[ I_k\colon R_i\leftarrow R_i+\lambda x_j\]
    \end{itemize}
    The time $t$ of the register program is the number of instructions $I_k$.
\end{definition}

\noindent
We note that we use the term input access when we are given direct access to the input $x$, while we use the term recursive call when we are designing a subroutine whose input is being received not from the global function but rather from some preceding intermediate computation. We will make this distinction clear where necessary.

With regards to such subroutines, it will be important to restrict our register programs to be of a form amenable to composition:

\begin{definition}
A register program is called \textit{clean} if, assuming all registers are initialized to some initial value, $\tau_i\in\mathbf{R}$, applying the sequence of instructions has the following effect:
\begin{itemize}
    \item 
    There is a subset of registers $S$ such that for all register $R_i\in S$, $R_i=\tau_i+ \delta_i$.
    \item 
    For every other register $R_j$, $R_j=\tau_j$.
\end{itemize}
\end{definition}

\noindent
The total time and space usage of such composed programs can thus be directly analyzed:

\begin{lemma}[Composition Lemma]
\label{composition_lemma}
    Let $f\colon \mathbf{R} \rightarrow \mathbf{R}$ and $g\colon \mathbf{R} \rightarrow \mathbf{R}$ be functions. Let $R_f$ and $R_g$ be clean register programs computing $f$ and $g$ with $t_f$ and $t_g$ recursive calls, $s_f$ and $s_g$ basic instructions, and $r_f$ and $r_g$ registers, respectively.
    Then there exists a register program $R_{f\circ g}$ that computes $f\circ g$ with
    \begin{itemize} 
        \item $t_ft_g$ recursive calls,
        \item $s_f + t_f s_g$ basic instructions, and
        \item $\max{\{r_f, r_g\}}$ registers.
    \end{itemize}
\end{lemma}
\begin{proof}
    Let $x\in \mathbf{R}$ and let $y=g(x)$, so $f(g(x))=f(y)$. We can cleanly compute $f(y)$ using $R_f$, with $t_f$ recursive calls to $y$, $s_f$ basic instructions, an input register and $r_f-1$ additional registers. On the other hand, we can cleanly compute $y=g(x)$ into the input register of $R_f$ using $t_g$ recursive calls and $s_g$ basic instructions for each of the $t_f$ recursive calls made to $g$.
    
    There are now two cases:
    \begin{itemize}
        \item 
        If $r_g-1 <r_f-1$, since $R_g$ cleanly computes $g$, we can use the  additional registers of $R_f$ to compute $y$.
        \item 
        If $r_g-1 >r_f-1$, we add $r_g-r_f$ registers and compute $y$ using the latter registers as well as the additional registers of $R_f$.
    \end{itemize}
    Thus we can compute $f(g(x))$ with $t_ft_g$ recursive calls to $x$ and $\max{\{r_f, r_g\}}$ registers.
\end{proof}

Lastly we connect clean register programs to catalytic computation:

\begin{lemma}[Lemma 15 in \cite{buhrman2014computing}]
\label{from_register_to_turing_machine}
    Any (uniform) clean register program of time $t$, space $s$, and with $n$ inputs over a finite ring $\mathbf{R}$ can be simulated by a catalytic Turing machine in pure space $O(\log t +\log n+ \log |\mathbf{R}|) $ and catalytic space $O(s\log |\mathbf{R}|)$.
\end{lemma}

\subsection{Polynomial Representation}
\label{polynomial_representation}
The question of representing Boolean functions as multivariate polynomials over a ring has been intensively studied and has proven to be a useful tool in circuit complexity (see \cite{beigel1993polynomial} for  a~survey). We will consider two kinds of representation: one that we will call the representation of $f$, and the other the weak representation of $f$. 

\begin{definition}
    Let $P$ be a polynomial on $n$ variables over a ring $\mathbf{R}$, and let $f$ be a Boolean function with $n$ inputs.
    We say $P$ \emph{represents} $f$ if for all inputs $x\in \{0,1\}^n$, we have $P(x)=0$ if and only if $f(x)=0$.
    We say $P$ \emph{weakly represents} $f$ if there exist two disjoint sets $S_0,S_1\subset \mathbf{R}$ such that for all inputs $x\in \{0,1\}^n$, we have $f(x)=b$ if and only if $P(x)\in S_b$, for $b\in\{0,1\}$.
\end{definition}
Let us underline the difference between these two definitions of representation with an example. Let us consider the $n$-ary AND function. Observe that the polynomial $P(X_1,\dots,X_n)=\sum_{i=1}^n X_i$ over any $\mathbb{Z}_m$, where $m>n$, weakly computes AND. Indeed, we can take $S_0=\{1,\dots,n-1\}$ and $S_1=\{n\}$. 
On the other hand, $P$ does not represent AND, and in general, it is known that AND cannot be represented by a degree $1$ polynomial.

\begin{definition}
    We will say that a Boolean function $f$ has a (weak) $(\mathbf{R},d)$-representation if there is a degree $d$ polynomial which (weakly) represents it. We also define $(\mathbf{R},d,s)$-representation, where additionally the polynomial is required to have at most $s$ monomials.
\end{definition}

\section{Register Programs for Polynomials}

\subsection{Computing Univariate Polynomials}

In order to prove $\TC^1\subseteq \CL$, \cite{buhrman2014computing} design a register program to compute $x^n$ for any element $x$ in a commutative field. We state a straightforward generalization of their lemma and corresponding program to compute arbitrary univariate polynomials:

\begin{lemma}
\label{univariate_polynomial}
Let $p\in \mathbb{N}$ be a prime number, and let $P\in \mathbb{F}_p[X]$ be a univariate polynomial of degree at most $n$. For all $x\in \mathbb{F}_p$, there is a register program that computes $P(x)$ with
\begin{itemize}
    \item $4$ recursive calls to $x$,
    \item $2n+2$ basic instructions, and
    \item $n+2$ registers.
\end{itemize}
\end{lemma}
\begin{proof}

Let $P=\sum_{i=0}^n a_i X^i$ for some coefficients $a_i$. Let $R_{in}$ be the input register initially equal to $\tau_{in}$. It is straightforward to show, by writing $x$ as $(\tau_{in} + x) - \tau_{in}$, that
\begin{equation*}
    P(x)= \sum_{i=0}^n a_i \sum_{j=0}^n \binom{i}{j} (\tau_{in}+x)^j(-\tau_{in})^{i-j} 
\end{equation*}
We will use this equation as well as the register program to compute $x^n$ in Lemma $4$ of \cite{buhrman2014computing}, to construct a register program to compute $P(x)$, given in \Cref{fig:program}.

\begin{figure}
    \begin{lstlisting}
Registers: 
$R_{in}= \tau_{in}$
$R_1=\tau_1,\dots,R_n=\tau_n$
$R_{out}=\tau_{out}$

$R_{in} \leftarrow R_{in}+x$ // $R_{in}=\tau_{in}+x$

For $1\leq i\leq n $
    $R_i \leftarrow R_{i}+ R_{in}^i$ // $R_i=\tau_i+(\tau_{in}+x)^i$
    
$R_{in} \leftarrow R_{in}-x$ // $R_{in}=\tau_{in}$

For $1\leq i\leq n $
    $R_{out} \leftarrow R_{out}+ a_i\left({(-1)}^{i}R_{in}^i+\sum_{j=1}^i \binom{i}{j} {(-1)}^{i-j} R_{j}R_{in}^{i-j}\right) $ 
// $R_{out}= \tau_{out}+ P(x)-a_0 + \sum_{i=1}^n a_i \sum_{j=1}^i \binom{i}{j} (-1)^{i-j} \tau_j\tau_{in}^{i-j}$

$R_{in} \leftarrow R_{in}+x$ // $R_{in}=\tau_{in}+x$
$R_i \leftarrow R_{i}- R_{in}^i$ // $R_i=\tau_i$
$R_{in} \leftarrow R_{in}-x$ // $R_{in}=\tau_{in}$

For $1\leq i\leq n $
    $R_{out} \leftarrow R_{out}- a_i\sum_{j=1}^i \binom{i}{j} (-1)^{i-j} R_{j}R_{in}^{i-j} $ 
// $R_{out}= \tau_{out}+ P(x)-a_0$

$R_{out} \leftarrow R_{out}+a_0$ // $R_{out}= \tau_{out}+ P(x)$
\end{lstlisting}
\caption{Program for computing a polynomial $P(x)$ of degree $n$ using $4$ recursive calls to $x$, $2n+2$ basic instructions, and $n+2$ registers}
\label{fig:program}
\end{figure}
\end{proof}

As discussed in \cite{buhrman2014computing,cook2023tree}, this program can be adapted to compute a set of polynomials $\{P_1,\dotsc, P_{\ell}\}$ with similar parameters:
\begin{lemma}
\label{univariate_polynomial_set}
Let $p\in \mathbb{N}$ be a prime number, and let $P_1 \ldots P_{\ell}\in \mathbb{F}_p[X]$ be univariates polynomial of degree at most $n$. For all $x\in \mathbb{F}_p$, there is a register program $\mathbf{UP}_{P_1,\dotsc, P_{\ell}}$ that computes $P_1(x) \ldots P_{\ell}(x)$ with
\begin{itemize}
    \item $4$ recursive calls to $x$,
    \item $2n+2\ell$ basic instructions, and
    \item $n+1+\ell$ registers.
\end{itemize}
\end{lemma}
\begin{proof}
    The program is the same as that of \Cref{univariate_polynomial}, but with each
    instruction involving $R_{out}$ replaced by $\ell$ instructions of the same form,
    where each involves a different output register $R_{out,j}$ and the corresponding
    polynomial $P_j(x)$.
\end{proof}

\subsection{Computing Multivariate Polynomials}
\label{polynomials_waring_rank}
In the last section we showed that computing powers of elements in a commutative ring can be done with a constant number of recursive calls. Moreover, it should be clear that computing the sum of variables is an easy operation: we simply add each variable to a fixed output register in turn (see~\cite{buhrman2014computing}).

Thus our goal is to represent our polynomial in such a way that we only have to perform addition and powering operations. Here our discussion of representations in \Cref{polynomial_representation} comes into play:

\begin{theorem}[\cite{schinzel2002decomposition,bialynicki2008representations}]
   \label{representation_linear_form}
   Let $P\in \mathbb{F}_p[x_1,\dotsc,x_n]$ be a homogeneous polynomial of degree $d<p$. There exist $m\in \mathbb{N}$, $m$ elements $\alpha_i$ and $nm$ elements $\beta_{i,j}$ such that:
\begin{equation*}
    P(x_1,\dots,x_n)= \sum_{i=1}^m \alpha_i{\left(\sum_{j=1}^n \beta_{i,j} x_j\right)^d}.
\end{equation*}
Moreover,  $m \leq \binom{n+d-1}{d-1}$.
\end{theorem}

This new representation does not involve any multiplication. We can thus describe a register program that computes a polynomial in $O(1)$ recursive calls:
\begin{lemma}
    \label{homogeneous_degree_d_linear_form}
       Let $P$ be a homogeneous degree $d<p$ polynomial $P(x_1,\dotsc,x_n)$ over $\mathbb{F}_p$. There is a register program $F_P$ that cleanly computes $P$ with
    \begin{itemize}
        \item 
         $4$ recursive calls to each $x_i$,
         \item 
         $O(d\binom{n+d}{d})$ basic instructions, and
         \item
         $O(d\binom{n+d}{d})$ registers over $\mathbb{F}_p$.
    \end{itemize}
\end{lemma}
\begin{proof}
    We use \Cref{representation_linear_form} and write $P$ as: 
\begin{equation*}
    P(x_1,\dots,x_n)= \sum_{i=1}^m \alpha_i\left(\sum_{j=1}^n \beta_{i,j} x_j\right)^d.
\end{equation*}
Let $f_i(x_1,\dots, x_n)= \sum_{j=1}^n \beta_{i,j} x_j $ and $g=x^d$. Then by the above discussion:
\begin{itemize}
    \item 
    $f_i$ can be cleanly computed with $1$ recursive call to each $x_i$, $0$ basic instructions, and $1$ register.
    \item 
    $g$ can be cleanly computed with $4$ recursive calls, $2d+2$ basic instructions, and $d+2$ registers using the register program from \Cref{univariate_polynomial}.
\end{itemize}
Hence by \Cref{composition_lemma}, there exists a register program $R_{i}$ that cleanly computes 

$g(f_i(x_1,\dots, x_n))= {(\sum_{j=1}^n \beta_{i,j} x_j)}^d$ with $4$ recursive calls, $2d+2$ basic instructions, and $d+2$ registers. 

We can compute all these programs $R_i$ in parallel, only sharing the output register $R_{out}$. This leaves us with a register program that cleanly computes $P$ in $R_{out}$ with: 
\begin{itemize}
    \item 
    $4$ recursive calls to each $x_i$,
    \item 
    $m(2d+2) \leq (2d+2)\binom{n+d-1}{d-1}+1$ basic instructions,
    and
    \item 
    $m(d+1)+1\leq (d+1)\binom{n+d-1}{d-1}+1$ registers,
\end{itemize}
as required in the statement of the lemma.
\end{proof}

This register program immediately generalizes to a non-homogeneous polynomial $P$ by considering the decomposition $P=\sum_{i=0}^d P_i$, where each $P_i$ is a homogeneous degree $i$ polynomial.

\begin{corollary}
    \label{general_small_degree_d_linear_form}
       Let $P$ be a degree $d<p$ polynomial $P(x_1,\dotsc,x_n)$ over $\mathbb{F}_p$. There is a register program $F_P$ that cleanly computes $P$ with
    \begin{itemize}
        \item 
         $4$ recursive calls to each $x_i$,
         \item
         $O(d^2\binom{n+d}{d})$ basic instructions, and
         \item
         $O(d^2\binom{n+d}{d})$ registers over $\mathbb{F}_p$.
    \end{itemize}
\end{corollary}

\noindent
This register program has the advantage of employing a constant number of recursive calls.
However, this method has two caveats. First, the cost in the number of registers is superpolynomial when $d$ is non-constant. Second, the degree $d$ is upper bounded by the field size.

Concerning the field issue, we will instead consider the field $\mathbb{F}_q$ for some $q>d$, such that $(P\bmod q)\bmod p= P\bmod p$. 
Observe that for any degree $d$ polynomial over $\mathbb{Z}$ where the coefficients are smaller than $p$, the polynomial evaluation for $x=(x_1,\dots,x_n)$, where $0\leq x_i\leq p$, is upper bounded by
\begin{equation*}
    P(x)\leq p \sum_{i=0}^d \binom{n}{i} p^i \leq 2^{n}p^{d+1}
\end{equation*}
Hence, we can first evaluate $P(x)$ over a field of size $q\geq2^{n}p^{d+1}$. This yields the most general register program, leaving yet the problem with the number of registers unresolved.
\begin{corollary}
    \label{general_degree_d_linear_form}
       Let $P$ be a degree $d$ polynomial $P(x_1,\dotsc,x_n)$ over $\mathbb{F}_p$. There is a register program that cleanly computes $P$ with
    \begin{itemize}
        \item 
         $4$ recursive calls to each $x_i$,
         \item
         $O(d^2\binom{n+d}{d})$ basic instructions,
         \item
         $O(d^2\binom{n+d}{d})$ registers over $\mathbb{F}_{q}$, where $q=O(2^{n}p^{d+1})$, and
         \item 
         $1$ register over $\mathbb{F}_p$.
    \end{itemize}
\end{corollary}

We also note one register program of orthogonal strength, namely greater in recursive calls but much lesser in space. In order to prove that Tree Evaluation is in $\SPACE(\log{n} \log\log{n})$ (later improved by Stoeckl, see~\cite{goldreich2024solving}), Cook and Mertz~\cite{cook2023tree} present a register program to compute multivariate polynomials, which we will also use later:
\begin{lemma}
\label{interpolation_register_program}
    Let $P(x_1,\dots,x_{n})$ be a polynomial of degree $d<p$ over a prime field $\mathbb{F}_p$. There exists a register program that computes $P$ with:
    \begin{itemize}
        \item 
        $n$ input registers
        \item 
        $1$ output register
        \item 
        $O(d)$ basic instructions 
        \item 
        $d+1$ instructions of the type $I_{\lambda,all}$ or its inverse $I^{-1}_{\lambda,all}$, where:
        \begin{align*}
        I_{\lambda,all} \colon \text{For }  &1\leq i \leq \delta,
        \\
        & R_{in,\ell} \leftarrow R_{in,\ell}+\lambda x_i.
    \end{align*}
    \end{itemize}
\end{lemma}

\subsection{Computing Boolean functions}
In the preceding section we presented a register program to compute polynomials over any finite field. However, for the case of polynomials over $\mathbb{F}_2$, i.e.\ Boolean functions, there are specific properties that we can exploit to make them simpler to compute, the most important being that any polynomial over $\mathbb{F}_2$ is multilinear since $x_i^2=x_i$.

Since any Boolean function $f$ is uniquely represented by a polynomial $P$ over $\mathbb{Z}_2$, we will directly consider Boolean functions in this section. Let us present a register program that computes any Boolean function $f$ given a representation of $f$. 

Let us first consider the case of symmetric functions. Given a symmetric function $f$ in $n$ variables, there is some univariate polynomial $g\colon [n]\mapsto \mathbb{N}$ such that $f(x_1,\dots, x_n)=g(x_1+\cdots+x_n)$. We can compute $g$ using \Cref{univariate_polynomial}, which yields the following register program:
\begin{lemma}
\label{bool_sym_const_rec_calls}
Let $n\in \mathbb{N}$, and let $f$ be a symmetric polynomial function.  Let $p\in\mathbb{N}$ be a prime number such that $p>n$. There is a register program that computes $f$ with:
 \begin{itemize}
            \item 
            $4$ recursive calls to each $x_i$,
            \item
            $O(n)$ basic instructions, and
            \item 
            $O(p)$ registers over $\mathbb{Z}_{p}$.
\end{itemize}
\end{lemma}

\noindent
From this we deduce the following lemma for polynomials in general:
\begin{lemma}
\label{representation_integer_boolean_function}
    Let $f\colon \{0,1\}^n \mapsto \{0,1\}$ be a boolean function which is $(\mathbb{Z},d,t)$-represented, and let $p$ be a prime number such that $p>\max\{d,t\}$. There is a register program which cleanly computes $f$ with:
    \begin{itemize}
        \item 
        $64$ recursive calls to each $x_i$,
        \item
        $O(tp^2 \log p)$ basic instructions, and
        \item 
        $O(t p)$ registers over $\mathbb{Z}_{p}$.
    \end{itemize}
\end{lemma}
\begin{proof}
    Let $P$ be the polynomial which represents $f$, which we write as a sum of terms
    \begin{equation*}
        P=\sum_{j=1}^t u_k
    \end{equation*}
    where each term $u_k$ has the form
    \begin{equation*}
       u_k = \prod_{j=1}^d x_{i_j}, \quad i_1,\dots,i_d\in[n]
    \end{equation*}
    Observe that each $u_k$ is symmetric, and so using the register program given by \Cref{bool_sym_const_rec_calls}, we deduce that we can compute all the terms in parallel with 
    \begin{itemize}
        \item 
        $4$ recursive calls to each $x_i$,
        \item 
        $O(tp)$ basic instructions, and
        \item 
        $O(t p)$ registers over $\mathbb{Z}_{p}$.
    \end{itemize}
    To compute $\sum_{k=1}^t u_k$, we can again use the register program of \Cref{bool_sym_const_rec_calls}.
    This yields a register program with 
    \begin{itemize}
        \item 
        $4$ recursive calls to each $u_j$,
        \item 
        $O(p)$ basic instructions, and
        \item 
        $O(p)$ registers over $\mathbb{Z}_{p}$.
    \end{itemize}
    Composing the two register programs using \Cref{composition_lemma}, we have a register program $R_f$ which computes $P$ with
    \begin{itemize}
        \item 
        $16$ recursive calls to each $x_i$,
        \item
        $O(tp^2)$ basic instructions, and
        \item 
        $O(tp)$ registers over $\mathbb{Z}_{p}$.
    \end{itemize}
    Lastly, we convert the computation of $P$ into a computation of $f$. Note that our final output register for $P$ is over $\mathbb{Z}_p$, which we represent in binary with $O(\log p)$ bits. We now apply the register program for the OR function from \Cref{univariate_polynomial} which uses $4$ recursive calls, $O(\log p)$ basic instructions, and $O(\log p)$ registers. Composing this with the rest of the program gives us a final program for computing $f$ with
    \begin{itemize}
        \item 
        $64$ recursive calls to each $x_i$,
        \item
        $O(tp^2 \log p)$ basic instructions, and
        \item 
        $O(tp)$ registers over $\mathbb{Z}_{p}$.
    \end{itemize}
    which completes the lemma.
\end{proof}

\subsection{Circuits via Merging Layers}
In this section, we will find efficient register programs for circuits, and from it efficient catalytic algorithms, by a strategy of merging layers and directly computing the ``super-functions'' that emerge. Our starting point is the following lemma, used in \cite{cleve1988computing,buhrman2014computing}.

\begin{lemma}
\label{depth_d_gates}
    Let $B$ be a set of Boolean functions such that for any function $g \in B$ we have a register program $P_g$ with at most $t$ recursive calls, $b$ basic instructions, and $r$ registers computing $g$. Let $C$ be a depth $d$, size $s$ circuit whose gates are functions in $B$. Then $C$ can be computed by a register program $P_C$ with
    \begin{itemize}
        \item $O(t^d)$ recursive calls,
        \item $O(sb \cdot t^d)$ instructions, and
        \item $O(rs)$ registers. 
    \end{itemize}
\end{lemma}
\begin{proof}
    We will perform the operations of each layer of the circuit in parallel. Each layer uses $t$ recursive calls to the previous layer and $sb$ basic instructions; hence for a depth $d$ circuit, applying \Cref{composition_lemma} iteratively gives $O(t^d)$ recursive calls to the last layer and $sb \cdot O(t^d)$ total basic instructions. Since each layer has at most $s$ gates, we will need at most $rs$ registers at each layer, and so again by \Cref{composition_lemma} this gives $2rs$ registers in total. 
\end{proof}

Our strategy will be as follows: instead of computing a height $h$ circuit $C$ with AND and OR gates layer by layer, we will show that we can compute $d$ layers at a time by an efficient register program, and thus consider the circuit $C'$ of height $h/d$ whose gates are themselves depth $d$ circuits. We do this using polynomial representations of such circuits, which we then combine with \Cref{representation_integer_boolean_function} to obtain the register programs in question. Our main statement is the following:

\begin{lemma}
\label{sac_circuit_poly}
    Let $C$ be a polynomial size depth $d$ circuit with fan-in $2$ AND gates and fan-in $\ell$ OR gates for some $\ell$. Then $C$ can be represented by a degree $k\leq 2^d$  polynomial with $t\leq \ell^{2^d}$ terms over $\mathbb{Z}$. 
\end{lemma}
\begin{proof}
    The proof is by induction. The claim for $d=0$ follows immediately since this is in this case the circuit computes one of the inputs. 

    We now assume that the claim holds for depth $d-1$, and let $C$ be a depth $d$ circuit. We apply the induction hypothesis to the children of the top gate $g$, and have two cases for $g$ itself:
    \begin{description}
        \item
        [If the top gate is an OR gate:] We can find a polynomial $P=\sum_{i=1}^\ell P_{i}$  representing the circuit $C$, where $P_i$ is the polynomial for each input of the OR gate. Using the induction hypothesis
        \begin{equation*}
            \deg P= \max_i{\deg P_{i}}\leq 2^{d-1}< 2^d.
        \end{equation*}
        Moreover, if we let $t_{i}$ be the number of terms of each $P_{i}$, the number of terms of $P$ is 
         \begin{equation*}
            t=\sum_i t_{i}\leq \ell \cdot \ell^{2^{d-1}} \leq \ell^{2^d}
        \end{equation*}
        \item[If the output gate is a binary AND gate:] Let $P_l$ and $P_r$  respectively be the polynomials representing the left and right children. The polynomial $P=P_lP_r$  represents the circuit $C$, and
       
        \begin{equation*}
            \deg P= \deg P_l+\deg P_{r}\leq 2^{d-1}+2^{d-1}=2^d.
        \end{equation*}
        On the other hand, 
        \begin{equation*}
            t=t_lt_r\leq \left(\ell^{2^{d-1}}\right)^2= \ell^{2^{d}}
        \end{equation*}
    \end{description}
    which completes the proof.
\end{proof}

Combining \Cref{sac_circuit_poly} for $\ell = n^{O(1)}$ with our register program for representations in \Cref{representation_integer_boolean_function}, we immediately obtain the following corollary.

\begin{corollary}
\label{representation_circuit_depth_d}
     Let $C$ be a size $s$ depth $d$ circuit on $n$ inputs with unbounded fan-in OR gates and fan-in $2$ AND gates, and let ${p}\in\mathbb{N}$ be a prime number such that $p > s^{2^d}$.
     There is a register program which cleanly computes $C$ with:
    \begin{itemize}
        \item 
        $64$ recursive calls to each $x_i$,
        \item 
        $s^{O(2^d)}$ basic instructions, and
        \item 
        $s^{2^d}$ registers over $\mathbb{Z}_{p}$.
    \end{itemize}
\end{corollary}

\noindent
Our main result follows for the right choice of $d$.

\begin{proof}[Proof of \Cref{main_sac}]
    Let $d\in \mathbb{N}$ be such that $d\leq \epsilon \log\log{n}$. \Cref{representation_circuit_depth_d} provides a register program cleanly computing any size $n^{O(1)}$ depth $d$ bounded fan-in OR fan-in $2$ AND circuit with
    \begin{itemize}
        \item 
        $64$ recursive calls to each $x_i$,
        \item 
        $2^{O(\log^{1+\epsilon}{n})}$ basic instructions, and
        \item 
        $2^{O(\log^{1+\epsilon}{n})}$ registers over $\mathbb{Z}_{p}$, where $p=2^{O(\log^{1+\epsilon}{n})}$.
    \end{itemize}
    Given an $\SACt$ circuit $C$ of size $n^{O(1)}$ and depth $O(\log^2{n})$, we will rewrite it as a circuit $C'$ of size at most $n^{O(1)}$ and depth $O(\frac{\log^2{n}}{d})$, where each gate is a size $n^{O(1)}$ depth $d$ circuit with unbounded fan-in OR and fan-in 2 AND gates. Hence using \Cref{depth_d_gates}, for $p=2^{O(\log^{1+\epsilon}{n})}$ we have a register program for $C$ with
     \begin{itemize}
        \item 
        $64^{O\left(\frac{\log^2{n}}{\epsilon\log\log{n}}\right)} = 2^{O\left(\frac{\log^2{n}}{\log\log{n}}\right)}$ recursive calls to each $x_i$,
        \item 
        $n^{O(1)} 2^{O(\log^{1+\epsilon} n)} \cdot 2^{O\left(\frac{\log^2{n}}{\log\log{n}}\right)} = 2^{O\left(\frac{\log^2{n}}{\log\log{n}}\right)}$ basic instructions, and
        \item 
        $n^{O(1)} \cdot 2^{O(\log^{1+\epsilon}{n})}=2^{O(\log^{1+\epsilon}{n})}$ registers over $\mathbb{Z}_{p}$.
    \end{itemize}
    At the end these recursive calls translate into basic instructions reading the input, giving a total time of $2^{O\left(\frac{\log^2{n}}{\log\log{n}}\right)}$. We can thus translate this register program to a catalytic machine using  \Cref{from_register_to_turing_machine}, and we deduce that:
    \begin{equation*}
    C \in \CSPACE\left(\frac{\log^2{n}}{\log\log{n}}, 2^{O(\log^{1+\epsilon}{n})}\right)
    \end{equation*}  
    as claimed.
\end{proof}
\subsection{Matrix Powering via Decomposition}
We now move to register programs for computing powers of a matrix $M\in M_n(\mathbb{Z}_p)$. A first attempt can be given by simply applying \Cref{homogeneous_degree_d_linear_form}, as computing $M^{d}$ is equivalent to computing $n^2$ degree $d$ polynomial in the $n^2$ coefficients; namely, if we denote by $m^{(d)}_{i,j}$ the coefficient of $M^d$, we have:
 \begin{equation*}
     m^{(d)}_{i,j}= \sum_{1\leq k_1,\dots, k_{d-1}\leq n} m_{i,k_1}\left(\prod_{i=1}^{d-2}m_{k_i,k_{i+1}}\right)m_{k_{d-1},j}.
 \end{equation*}

We can therefore use \Cref{homogeneous_degree_d_linear_form} to compute $M^d$ for $d<p$:

\begin{lemma}
\label{lemma_powering_small_d}
    Let $M\in M_n(\mathbb{F}_p)$. There is a register program that computes $M^d$ for $d<p$
    with:
    \begin{itemize}
        \item 
        $4$ recursive calls to $M$, 
        \item 
        $O(dn^2\binom{n+d}{d})$ basic instructions, and
        \item
        $O\left(dn^2\binom{n^2+d}{d}\right)$ registers over $\mathbb{F}_p$.
\end{itemize}
\end{lemma}

The register program from \Cref{lemma_powering_small_d} is a first step towards a more generic program for matrix powering, but it has two major issues:
\begin{itemize}
    \item 
    it works only to compute powers up to $p-1$.
    \item 
    the number of registers grows exponentially with $d$.
\end{itemize}
We address these issues independently to get a register program which can handle all the cases.

\paragraph*{Computing any power $d$}

Let $M\in M_n(\mathbb{F}_p)$ and let $L\in M_n(\mathbb{Z})$ be the natural extension of $M$ to integers.
Observe that the coefficients $\ell^{(d)}_{i,j}$ of $L^d$ are evaluations of degree $d$ polynomials with $n^{d-1}$ terms, and hence $\ell^{(d)}_{i,j}\leq p^d n^{d-1}$.

Let us consider the first prime number $q$ 
greater than $p^d n^{d-1}\geq d$. In this case, we have $\left(\ell^{(d)}_{i,j} \bmod q \right)\bmod p= \ell^{(d)}_{i,j}\bmod p = m^{(d)}_{i,j}$. Hence we can use the register program of \Cref{lemma_powering_small_d} for $d<q$ and have an output register over $\mathbb{F}_p$ for each coefficient. This yields the following:

\begin{lemma}
\label{lemma_powering_all_n}
    Let $M\in M_n(\mathbb{F}_p)$. There is a register program that computes $M^d$ with:
    \begin{itemize}
        \item 
        $4$ recursive calls to $M$, 
        \item 
        $O\left(dn^2\binom{n^2+d}{d}\right)$ basic instructions,
        \item
        $O\left(dn^2\binom{n^2+d}{d}\right)$ registers over  $\mathbb{F}_{O(p^d n^{d-1})}$, and
        \item 
        $O(n^2)$ registers over $\mathbb{F}_p$.
    \end{itemize}
\end{lemma}

\paragraph*{Reducing the Number of Registers}
The latter register program works for all $d$ and has a constant number of recursive calls. However, the number of registers is still exponential in $d$, and therefore it is not usable as is. To fix this issue, let us observe that, if we let $f_d$ be the powering function $f_d(M)=M^d$, we have $f_{d^k}=(f_d)^{\circ^k }$, where $\circ^k$ means composing $k$ times the same function. An iterated application of \Cref{composition_lemma} gives the following:




\begin{lemma}
\label{boosting_matrix_powering}
    Let $\mathbf{R}$ be a ring, $x\in \mathbf{R}$ and $\delta\in\mathbb{N}$. Suppose that for all $k\leq \delta$ there is a clean register program $P_k$ computing $x^k$ with at most $t$ recursive calls, $s$ basic instructions and $r$ registers. 
    Then, there exists a register program $P$ that computes $x^d$ with
    \begin{itemize}
        \item at most $(\lceil \log_\delta{d}\rceil+2) \frac{t}{t-1} t^{\lceil\log_\delta{d}\rceil}$ recursive calls,
        \item 
        $O(\log_\delta{d}) \frac{t+s(t+1)^{\lceil\log_\delta{d}\rceil}}{t} $ basic instructions, and
        \item $1+r(\lfloor \log_\delta{d}\rfloor+1)$ registers over $\mathbf{R}$.
    \end{itemize}
\end{lemma}
\begin{proof}
    Let $f_k\colon \mathbf{R} \rightarrow \mathbf{R}$ be the functions that computes $x^k$ for all $k\leq d$. 
    Observe that 
    $$d=\sum_{i=0}^{\lfloor \log_\delta{d}\rfloor} \alpha_i \delta^i$$
    where $\alpha_i$ are non-negative integers smaller than $\delta$. Therefore
    \begin{equation*}
        x^d= \prod_{i=0}^{\lfloor \log_\delta{d}\rfloor} (x^{\delta^i})^{\alpha_i}
    \end{equation*}
    Note that $ (x^{\delta^i})^{\alpha_i}=f_{\alpha_i}\circ {(f_\delta)}^{i}(x)$. We can apply \Cref{composition_lemma} to obtain a register program $R_i$ that cleanly computes $(x^{\delta^i})^{\alpha_i}$ with $t^{i+1}$ recursive calls, $(t+1)^i s$ basic instructions, and $r$ registers.
    Then, using \Cref{interpolation_register_program} for $P = \prod_i (x^{\delta^i})^{\alpha_i}$,
%
    %
    we can compute $x^d$ by calling $\lfloor \log_\delta{d}\rfloor+2\leq \lceil \log_\delta{d}\rceil+2$ times each program $R_i$,  and using a single additional register. In total, we require
    \begin{itemize}
        \item 
        at most $(\lceil \log_\delta{d}\rceil+2)\sum_{i=0}^{\lfloor \log_\delta{d}\rfloor} t^{i+1} \leq (\lceil \log_\delta{d}\rceil+2) \frac{t}{t-1} t^{\lceil\log_\delta{d}\rceil}  $ recursive calls to $x$,
        \item 
        $O(\log_\delta{d})\left(1+s\sum_{i=0}^{\lfloor \log_\delta{d}\rfloor} (t+1)^{i}\right) \leq O(\log_\delta{d}) \frac{t+s(t+1)^{\lceil\log_\delta{d}\rceil}}{t} $ basic instructions, and
        \item 
        $1+\sum_{i=0}^{\lfloor \log_\delta{d}\rfloor} r =1+r(\lfloor \log_\delta{d}\rfloor+1)$ registers. \qedhere
    \end{itemize} 
\end{proof}

\noindent
Hence, we fix some $\delta \in \mathbb{N}$, and we can compute $M^d$ for all $d$ based on the register program for $\delta$. This yields our register program for \Cref{Matrix_powering_intro}.

\begin{proof}[Proof of \Cref{Matrix_powering_intro}]
    Let $\delta\in \mathbb{N}$. Combining \Cref{lemma_powering_all_n} and \Cref{boosting_matrix_powering}, we get that for all $d\in\mathbb{N}$, there exists a register program that computes $M^d$ with: 
    \begin{itemize}
        \item 
        $O(\frac{\log{d}}{\log{\delta}}d^{\frac{3}{\log{\delta}}})$ recursive calls to $M$,
        \item 
        $O\left(\frac{\delta n^2d^{\frac{3}{\log{\delta}}}\log{d}}{\log{\delta}}\binom{n^2+\delta}{\delta}\right)$ basic instructions,
        \item 
        $O(\frac{n^2 \log{d}}{\log{\delta}})$ registers over $\mathbb{F}_p$.
    \end{itemize}
    Replacing by $\epsilon=\frac{3}{\log{\delta}} $, yields a program with
    \begin{itemize}
        \item 
         $O(\epsilon d^{\epsilon}\log{d})$ recursive calls to $M$,
         \item 
        $O\left(\epsilon \frac{n^{2^{\frac{3}{\epsilon}+1}}}{2^{\frac{3}{\epsilon}-1}}d^{\epsilon}\log{d}\right)$ basic instructions,
        \item
          $O\left(\epsilon \frac{n^{2^{\frac{3}{\epsilon}+1}}}{2^{\frac{3}{\epsilon}-1}}\log{d}\right)$ registers over  $\mathbb{F}_q$ for $q = O\left(\left(np\right)^{2^{\frac{3}{\epsilon}}} \right)$, and
        \item 
        $O(\epsilon n^2\log{d})$ registers over $\mathbb{F}_p$.
    \end{itemize}
    which completes the proof.
\end{proof}

\bibliographystyle{alpha}
\bibliography{bibliography}

\end{document}